\documentclass{article}

\usepackage{latexsym}
\usepackage{enumerate, amsmath, amssymb}
\usepackage{amsthm}

\usepackage[mathscr]{eucal}

\newcommand{\Real}{\mathbb R}

\newcommand{\R}{\mathbb{R}}

\newtheorem{Proposition}{Proposition}
\newtheorem{Lemma}{Lemma}
\newtheorem{Remark}{Remark}

\newtheorem{Definition}{Definition}

\newtheorem{Example}{Example}

\title{Dispersive behavior in Galactic Dynamics}
\author{Simone Calogero, Juan Calvo,\\ \'Oscar S\'anchez \& Juan Soler\\[1cm]
Departamento de Matem\'atica Aplicada,\\
Facultad de Ciencias, Universidad de Granada,\\
18071 Granada, Spain\footnote{{\tt E-mail: calogero@ugr.es, juancalvo@ugr.es, ossanche@ugr.es, jsoler@ugr.es}}} 
\date{}
\begin{document}
\maketitle
\begin{abstract}{The purpose of this paper is to study the relations between different concepts of dispersive solution for the Vlasov-Poisson system in the gravitational case. Moreover we give necessary conditions for the existence of partially and totally dispersive solutions and a sufficient condition for the occurence of statistical dispersion. These conditions take the form of inequalities involving the energy, the mass and the momentum of the solution. Examples of dispersive and non-dispersive solutions---steady states, periodic solutions and virialized solutions---are also considered.}
\end{abstract}

%%%%%%%%%%%%%%%%%%%%%%%%%
\noindent{\small\bf AMS classification (2000). }{\scriptsize Primary: 35B05, 35B40, 82B40, 82C40}
%%%%%%%%%%%%%%%%%%%%%%%%%

%%%%%%%%%%%%%%%%%%%%%%%%%
\section{Introduction}\label{introduction}
%%%%%%%%%%%%%%%%%%%%%%%%%%

%The aim of this paper is to analyze the time dispersion properties of solutions to the Vlasov-Poisson system in the gravitational case and to investigate under which conditions steady states (corresponding to equilibrium configurations of the system) and time periodic solutions (breathers) may exist. 

Gravitational systems composed by a large number of classical particles in which collisions and external forces are negligible, {\it e.g.} the stars of a galaxy, can be described by the Vlasov-Poisson system. The regularity of solutions to this system has been extensively studied in the mathematical literature and this problem is by now well-understood, cf. \cite{LP, Pf, Sch}. On the contrary, very little is known on the time asymptotics of the Vlasov-Poisson system.

The large time behavior of solutions to the Vlasov-Poisson system is relatively simple in the case of electrostatic interaction among the particles, which is obtained from the gravitational case by reversing the sign in the right hand side of the field (Poisson) equation. In the electrostastic case all solutions exhibit a (strong or $L^q$-norm) dispersive character~\cite{IR,Per}. 
In the gravitational case the dynamics is more intricate: There exist (stable or unstable) steady states, periodic solutions (breathers) and (fully or partially) dispersive solutions. 
%In the case of the Einstein-Vlasov system, one also have to add the possible formation of black holes to this scenario. 
For the applications in Astrophysics it would be desirable to have a classification of the possible asymptotic behavior of solutions in terms of relations between quantities preserved by the evolution (such us the energy and the mass). This is clearly a very difficult---may be impossible---task, but in this paper we show that partial answers in this direction can be given for the Vlasov-Poisson system. We shall focus in particular on solutions that for large times exhibit some sort of dispersive behavior. The examples considered in the following sections show that solutions of the Vlasov-Poisson system may disperse in several different ways. Another goal of the present paper is to study the relations between these various concepts of dispersion.  
%{\bf In order to ilustrate the differences between these concepts we provide exact solutions as examples of each of them.}

Let us recall some fundamental facts about the Vlasov-Poisson system in the gravitational (galactic dynamics) case (more details can be found in \cite{ReinLibro}). The dynamics of the stars of the galaxy is described by the distribution function in phase space $f=f(t,x,p)$, which gives the probability density to find a particle (star) at time $t\in\R$ in the position $x\in\R^3$ with momentum $p\in\R^3$. The mass density $\rho=\rho(t,x)$ of the galaxy is given by
\begin{equation}\label{rhoVP}
\rho=\int_{\R^3} f\,dp\:.
\end{equation}
For notational simplicity we assume that the stars have all the same mass $m$ and fixed units such that $m=4\pi G=1$, where $G$ is Newton's gravitational constant.
The gravitational potential $U=U(t,x)$ generated by the galaxy solves the Poisson equation
\begin{equation}\label{poisson}
\Delta_x U=\rho\:,\quad\lim_{|x|\to\infty}U=0\:,\ \forall\,t\in\R\:,
\end{equation}
where the boundary condition  at infinity means that the galaxy is isolated. The assumption that the stars interact only by gravity leads to the Vlasov equation:
\begin{equation}\label{vlasovVP}
\partial_t f +p\cdot\nabla_x f -\nabla _xU\cdot\nabla_p f=0\:.
\end{equation}
The system~\eqref{rhoVP}--\eqref{vlasovVP} is the Vlasov-Poisson system. The solution of~\eqref{poisson} is given by the formula 
\begin{equation}\label{potential}
U=-\frac{1}{4\pi}\int_{\R^3}\frac{\rho(t,y)}{|x-y|}\,dy\:, 
\end{equation}
whence the Vlasov-Poisson system is equivalent to the non-linear Vlasov equation obtained by replacing the formula for $U$ in~\eqref{vlasovVP}. The energy $E$ and the mass $M$ of a solution are given by 
\begin{equation}\label{energymassVP}
E=\frac{1}{2}\int_{\R^6}|p|^2f\,dp\,dx-\frac{1}{2}\int_{\R^3}|\nabla_x U|^2dx\:,\quad M=\int_{\R^6}f\,dp\,dx
\end{equation}
and are conserved quantities. Likewise, the total linear momentum $Q$ and angular momentum $L$,
\begin{equation}
Q=\int_{\R^6} p\,f\,dp\,dx\:,\qquad L=\int_{\R^6}x\wedge p\,f\,dp\,dx\:,
\end{equation}
are conserved quantities. Moreover, since the characteristic flow of the Vlasov equation preserves the Lebesgue measure, then all $L^q$ norms of $f$ are preserved:
\begin{equation}\label{consLq}
\|f(t)\|_{L^q}=\mathrm{const.}\:,\quad\text{for all }1\leq q\leq\infty\:.
\end{equation}
Hereafter we denote by $L^q$ either $L^q(\R^6)$ or $L^q(\R^3)$, depending on the function under consideration, {\it e.g.} $f(t) \equiv f(t,\cdot,\cdot)$ or $\rho(t) \equiv \rho(t,\cdot)$. 
The invariance of Vlasov-Poisson by (time dependent) Galilean transformations is the property that, given $u\in\R^3$ and the transformation of coordinates
\begin{equation}\label{galileo}
\mathcal{G}_u:\quad t'=t\:,\quad x'=x-ut\:,\quad p'=p-u\:,
\end{equation}
then $f_u(t,x,p)=f(t',x',p')$ and $U_u(t,x)=U(t',x')$ solve the system~\eqref{rhoVP}--\eqref{vlasovVP} if and only if $(f,U)$ does. Note that $Q$ can be made to vanish with a suitable Galilean transformation; the resulting reference frame is at rest with respect to the center of mass of the distribution, which is defined as
\begin{equation}\label{centerofmass}
c_\rho(t)=M^{-1}\int_{\R^3} x\,\rho\,dx\:.
\end{equation} 
Throughout the paper we assume that $f$ is a non-trivial global classical solution of the Vlasov-Poisson system such that, at any fixed time $t$, $f$ has compact support in the variables $(x,p)$ (however this assumption can be substituted by suitable decay conditions on the variables $(x,p)$ or by requiring only that $f$ has bounded moments in these variables up to a sufficientely high order).  We shall refer to these solutions as {\it regular solutions}. It is well known that for any initial datum $0\leq f^0=f_{|t=0}\in C^1_c(\R^6)$, there exists a unique global regular solution of the Vlasov-Poisson system, see \cite{LP, Pf, Sch}.

An important role in our discussion is played by spherically symmetric solutions of the Vlasov-Poisson system. A solution of Vlasov-Poisson is spherically symmetric if $f(t,Ax,Ap)=f(t,x,p)$ for all rotations $A\in\mathrm{SO(3)}$.  The potential induced by a spherically symmetric solution is a function of the radial variable $r=|x|$ only and indeed we have the representation formula
\begin{equation}\label{drU}
\partial_r U=\frac{1}{r^2}\int_0^r\lambda^2\rho(t,\lambda)\,d\lambda\:.
\end{equation}
Clearly, the center of mass of spherically symmetric solutions is at $r=0$.

This paper is organized as follows. In Section~\ref{dispdef} we introduce several concepts of dispersion for a mass distribution---not necessarily originated by the Vlasov-Poisson system and give some examples. In Section~\ref{applVP} we specialize to the case of the Vlasov-Poisson system. We give necessary or sufficient conditions for the existence of various types of dispersive solution. For completeness we also briefly discuss in Section~\ref{nondispersive} two examples of non-dispersive solutions, namely periodic solutions and steady states. 
%Section~\ref{conclusions} contains some concluding remarks and open problems. 

To conclude this Introduction we remark that the Vlasov-Poisson system ceases to be valid as a physical model when the stars move with large velocities (of the order of the speed of light) or in the presence of very massive galaxies, since then relativistic effects become important. Typical relativistic effects are the redshift of the luminous signals emitted by the galaxy and the formation of black holes. The model which is currentely believed to represent the physically correct relativistic generalization of the Vlasov-Poisson system is the Einstein-Vlasov system~\cite{Hakan}, where Poisson's equation is substituted by Einstein's equations of General Relativity. 
As compared to the Vlasov-Poisson system, the Einstein-Vlasov system is far more complicated and less understood. 
Another relativistic generalization of Vlasov-Poisson is the Nordstr\"om-Vlasov system~\cite{Calogero}, which, although not physically correct, is mathematically interesting since it already captures some of the technical and conceptual new difficulties that are encountered when studying a relativistic ({\it i.e.}, Lorentz invariant) system. The large time behavior of the relativistic models will be discussed elsewhere~\cite{unp}.

%{\bf Si no metemos nada de relatividad este p\'arrafo no tiene mucho sentido. Y aparte podr\'{\i}amos cambiar la notaci\'on $(x,p)$ por $(x,v)$}. We remark that the dispersive character of solutions for a large class of PDEs can be related to the properties of the symmetry group of an Action Functional associated to the system,
%or more concretely to the lack of invariance with respect to some specific transformation. In \cite{RAS} it was proved that the variation of the Action Functional for the Schr\"odinger-Poisson system under dilatations and under conformal transformations implies a dispersion type equation. This general fact is also true for non-linear wave equations, see \cite{St}, and will play a fundamental role in the present paper.

%%%%%%%%%%%%%%%%%%%%%%%%%%%%%%%%%%
\section{Dispersive behavior}\label{dispdef}
%%%%%%%%%%%%%%%%%%%%%%%%%%%%%%%%%%
%In the following we discuss about what do we mean when we say that a solution of the Vlasov-Poisson system is dispersive: many different concepts can be introduced, but here we will focus on some of them that are defined only in terms of the associated mass distribution (which evolves in time). These have the potential to be generalized to broader settings (different kinetic models, for instance). So we shall propose some definitions of dispersive solutions in a general setting and for each of them we give an example of solution to the Vlasov-Poisson system which illustrates the kind of dispersive behavior they refer to. 
In this section we introduce several concepts of dispersion for regular mass distributions. By a {\it regular mass distribution} of total mass $M$ we mean a non-negative $C^1$ function $\rho(t,x)$ such that $\rho(t,\cdot)$ has compact support  and  $\|\rho(t)\|_1=M$ (independent of time $t$). This terminology is consistent with the one used for solutions of the Vlasov-Poisson system: The mass distribution $\rho$ defined by~\eqref{rhoVP} is regular whenever $f$ is regular. 
\subsection{Strong dispersion}
\begin{Definition}
A regular mass distribution $\rho$ is said to be \textnormal{strongly} dispersive if there exists $q>1$ such that the limit
\begin{equation}\label{strong}
\lim_{t\to\infty}\|\rho(t)\|_{L^q}\ \text{exists and is zero}\:.
\end{equation}
\end{Definition}
Obviously, strong dispersion is a Galilean invariant concept. For the Vlasov-Poisson system in the plasma physics case, which is obtained from~\eqref{rhoVP}--\eqref{vlasovVP} by reversing the sign in the right hand side of~\eqref{poisson}, it was proved in \cite{IR, Per} (see also \cite{GS}) that {\it all} solutions are strongly dispersive, with~\eqref{strong} being verified for $q\in (1,5/3]$. In the gravitational case, examples of strongly dispersive solutions are those constructed in \cite{Bardos-Degond} for small initial data, see also \cite{Hwang}. For these solutions there holds the estimate
\[
\rho(t,x)\leq C(1+|t|+|x|)^{-3}\:,
\]
for a positive constant $C$, which clearly implies strong dispersion. 

\subsection{Total and partial dispersion}
The next types of dispersive solution that we are going to discuss use the notion of `concentration function of a measure' introduced by L\'evy~\cite{Levy} and applied by P.-L. Lions in the proof of the concentration-compactness Lemma~\cite{Co-Co}. 
%and observe that
%\begin{equation}\label{M=0eq}
%\mathcal{M}_\infty=0\quad\text{ iff }\quad\lim_{t\to\infty}\sup_{x_0\in\R^3}\int_{|x-x_0|<R}\rho(t,x)\,dx=0\,,\ \forall R>0\:.
%\end{equation}
\begin{Definition}\label{dispersion}
A regular mass distribution $\rho$ is said to be \textnormal{totally}, respectively \textnormal{partially} dispersive, if and only the limit
\begin{equation}\label{MR}
\mathcal{M}(R)=\lim_{t\to\infty}\sup_{x_0\in\R^3}\int_{|x-x_0|<R}\rho(t,x)\,dx\:,
\end{equation}
exists and
\begin{equation}\label{massloss}
\mathcal{M}_\infty=\lim_{R\to\infty}\mathcal{M}(R)
\end{equation}
satisfies $\mathcal{M}_\infty=0$, respectively $\mathcal{M}_\infty\in (0,M)$.

\end{Definition}
%{\bf Is every totally dispersive solution strongly dispersive also?}
%\begin{Remark}\textnormal{
%{\bf Yo extender\'{\i}a la anterior definici\'on al caso de una sucesi\'on cualquiera de tiempos.}} 
%\end{Remark}

\begin{Remark}\textnormal{
Of course, it is possible that $\mathcal{M}(R)$ could not be well defined for all $R$ ({\it e.g.} when $\rho(t)$ is time periodic). Whenever it exists, $\mathcal{M}(R)$ is a bounded non-decreasing function and therefore the limit~\eqref{massloss} is well defined. Moreover
$\mathcal{M}_\infty\in [0,M]$.} 
\end{Remark}

It is clear that strong dispersion implies total dispersion. Moreover, 
%by~\eqref{M=0eq}, 
total dispersion is equivalent to  the vanishing property in the concentration-compactness theory, see~\cite{Co-Co}; precisely, a mass distribution $\rho$ is totally dispersive if and only if the limit
\begin{eqnarray}\label{mgs}
\lim_{t \rightarrow \infty}\sup_{x_0\in\R^3}  \int_{|x -x_0|<R} \rho(t,x) \,dx\ \textnormal{exists and is zero, }\forall\,R>0\:.
\end{eqnarray}
An important physical property of~\eqref{mgs} is that it is invariant by Galilean transformations, unlike the decay of the mass (or energy) over a ball with arbitrary radius,  
\begin{eqnarray}\label{localmassdecay}
\int_{|x|\leq R} \rho \,dx \to 0\:,\quad\text{as }\ t\to\infty\:,\quad\forall R>0\:,
\end{eqnarray}
which has also been used as definition of dispersion for evolution type equations (including non-linear Vlasov equations), see~\cite{GS, St} for instance. Thus for example, according to our definition of total dispersion, a static ({\it i.e.}, time independent) solution which is ``put in motion" by a Galilean transformation is not to be regarded as a dispersive solution (whereas it would be so according to~\eqref{localmassdecay}).

\begin{Example}\label{shellexample}\textnormal{
As an example of totally dispersive solution of the Vlasov-Poisson system, consider a spherically symmetric shell of matter with internal radius $R_1(t)$ and---possibly infinite---external radius $R_2(t)$ (this example was first introduced in \cite{Andreasson2008}). Let $r=|x|$ and $w=x\cdot p/r$, the radial momentum variable. Now suppose that in the support of %$f^0=f(0,x,p)$, 
$f^0=f(0,x,p)$,
it is verified that 
\begin{equation}\label{initialw}
\inf\{w,w\in\textrm{ supp } f^0\}>\sqrt{\frac{M}{2\pi R_1(0)}}\:,
\end{equation}
where $M$ is the total mass, {\it i.e.}, initially the particles are moving outwardly with sufficiently high speed. Using that in spherical symmetry the maximal force experienced by a particle is bounded by $M/4\pi r^2$, see~\eqref{drU}, we find that, along the characteristics of the Vlasov equation,
\[
\frac{d}{dt}\left(\frac{1}{2}w^2-\frac{M}{4 \pi r}\right)=w\,\dot{w}+\frac{M}{4 \pi r^2}\,\dot{r}=w\left(\dot{w}+\frac{M}{4 \pi r^2}\right),
\]
which is positive in the time interval $[0,T)$ in which $w>0$, {\it i.e.}, as long as the shell keeps moving outwardly. It follows that 
\[
w(t)^2>w(0)^2-\frac{M}{2\pi r(0)}>\inf_{\textrm{supp}\ f^0}w^2-\frac{M}{2 \pi R_1(0)}:=W>0\:,
\]
where for the last inequality we use~\eqref{initialw}. This implies that $T=\infty$, that is, the shell moves outwardly for all future times. Moreover, $W>0$ is a uniform lower bound on the radial momentum, which entails
\begin{equation}\label{Rincrease}
R_2(t)>R_1(t)>R_1(0)+Wt\:.
\end{equation}
We claim that, because of~\eqref{Rincrease}, the solution under consideration is totally dispersive. 
We shall achieve this by proving that the potential energy vanishes in the limit $t\to\infty$, which for a solution of Vlasov-Poisson is equivalent to total dispersion, see Proposition~\ref{fullnec} in the next section. Thanks to the rotational symmetry we have the representation~\eqref{drU},
which allows to estimate the potential energy as
\begin{align*}
E_{\mathrm{pot}}(t)&=\frac{1}{2}\int_{\R^3}|\nabla_x U|^2dx=\frac{1}{2} \int_{\R^3 } \left( \frac{1}{r^2} \int_0^r \lambda^2 \rho(t,\lambda) \ d\lambda \right)^2 \ dx\\
& = 2 \pi \int_{R_1(t)}^{R_2(t)} \frac{1}{r^2} \left(\int_0^r \lambda^2 \rho(t,\lambda) \ d\lambda \right)^2 \ dr\\
&\leq\frac{1}{8 \pi} \int_{R_1(t)}^{R_2(t)} \frac{M^2}{r^2}\ dr \le \frac{M^2}{8 \pi R_1(t)}
\end{align*}
and the claim follows.}
\end{Example}

%%%
For the next result we denote by $d(A,B)$ the distance of the sets $A,B\subset\R^3$:
\[
d(A,B)=\inf\{|x-y|\:, x\in A\;,y\in B\}\:.
\]
\begin{Lemma}\label{dichotomy} Let $\rho$ be a partially dispersive regular mass distribution. Then, for any given $\varepsilon >0$ there exist  $t_n\xrightarrow{n}\infty$ and two sequences of non-negative $L^1$ functions $\rho_1^n,\rho_2^n:\R^3\to\R_+$,  such that $\rho(t_n)\geq \rho_1^n+\rho_2^n$ and
\begin{itemize}
\item[a)] $\|\rho(t_n)-(\rho_1^n+\rho_2^n)\|_1 \leq \varepsilon$;
\item[b)] $\left|\|\rho_1^n\|_1 - \mathcal{M}_\infty \right| \leq \varepsilon \:,\left| \|\rho_2^n\|_1 - (M-\mathcal{M}_\infty)\right|   \leq \varepsilon$;
\item[c)] $d(\mathrm{supp}\rho_1^n,\mathrm{supp}\rho_2^n)\to\infty$, as $n\to\infty$.
\item[d)] There exists a sequence of vectors $y^n\in\R^3$ and $0<R^*_{(\varepsilon)}$ such that $\rho_1^n=0$, for $|x-y^n|> R^*_{(\varepsilon)}$.
\end{itemize}
\end{Lemma}
\begin{Remark}\textnormal{
The conditions a)-c) define the {\it dichotomy} property of the mass distribution $\rho$ in the concentration-compactness Lemma, see~\cite{Co-Co}. Condition $d)$ is also consequence of the same result, although this was not pointed out in~\cite{Co-Co}; for the sake of completeness we shall give here the proof of Lemma~\ref{dichotomy}. In our context, the relevance of the extra condition $d)$ arises from the fact that it prevents the system from being strongly dispersive, as we will see in Section \ref{kurthsec}.}
\label{rem2}
 \end{Remark}
\begin{proof}
The following proof is adapted from~\cite{Co-Co}. Owing to (\ref{mgs}), along any sequence $t_n\xrightarrow{n}\infty$ we have
\[
\lim_{n\to\infty}\sup_{x_0\in\R^3}\int_{|x-x_0|<R}\rho(t_n,x)\,dx=\mathcal{M}(R)\:.
\]
Since $\mathcal{M}(R)\to\mathcal{M}_\infty\in (0,M)$,  for all $\varepsilon>0$ we can find $R^*=R_{(\varepsilon)}^*$ such that, for all $n$ sufficiently large, 
\[
\sup_{x_0\in\R^3}\int_{|x-x_0|<R*}\rho(t_n,x)\,dx\in (\mathcal{M}_\infty-\varepsilon,\mathcal{M}_\infty+\varepsilon)\:.
\]
Moreover there exists $y^n\in\R^3$ such that
\[
\int_{|x-y^n|<R^*}\rho(t_n,x)\,dx\in (\mathcal{M}_\infty-\varepsilon,\mathcal{M}_\infty+\varepsilon)\:.
\]
Finally, we can find a sequence $R_n\xrightarrow{n}\infty$ and a subsequence of times---still denoted $t_n$---such that
\[
\sup_{x_0\in\R^3}\int_{|x-x_0|<R_n}\rho(t_n,x)\,dx\in (\mathcal{M}_\infty-\varepsilon,\mathcal{M}_\infty+\varepsilon)\:;
\]
The functions $\rho_1^n=\rho(t_n)\chi_{\{B_{y^n}(R*)\}}$ and  $\rho_2^n=\rho(t_n)\chi_{\{\R^3\setminus B_{y^n}(R_n)\}}$ are easily seen to satisfy the properties a)-d).
\end{proof}

%\begin{Remark}[prescindible]
%We have examples of solutions which are, in the sense of P.L.Lions, dichotomic and vanishing at the same time. The easiest one: take the ({\bf forthcoming}) example of the shell with something inside and this time put there the Kurth solution which disperses ``slowly''.
%\end{Remark}

\subsection{Statistical dispersion}
We shall now discuss another Galilean invariant concept of dispersion, which was introduced in \cite{DSS}. Define the statistical dispersion operator in space by
$$
\langle (\Delta x)^2 \rangle \::=\frac{1}{M}\left[ \int_{\Real^3} |x|^2 \rho(t,x) \; dx - \frac{1}{M}\left(\int_{\Real^3} x \rho(t,x) \; dx\right)^2\right]\:.
$$
%and
%$$
%\langle (\Delta v)^2 \rangle \quad := \int_{\Real^6}|v|^2 f(t,x,v) \; dx\,
%dv - \left(\int_{\Real^6} v f(t,x,v) \; dx\, dv\right)^2\, .
%$$
Up to a mass normalization, the statistical dispersion operator 
coincides with the statistical variance of the density mass function
and, consequently, it is a measure of the dispersion of such 
distribution. Note also that
\[
\langle (\Delta x)^2 \rangle=\int_{\R^3}\left|x-c_\rho(t)\right|^2\,\frac{\rho}{M}\,dx\:
\]
and therefore the statistical dispersion operator coincides with the moment of inertia of the mass distribution with respect to the center of mass. 
\begin{Definition}
A regular mass distribution $\rho$ is said to be \textnormal{statistically} dispersive if and only if 
\[
\sup_{t>0}\langle(\Delta x)^2\rangle\ =+\infty\:.
\]
\end{Definition}
\begin{Remark}
\textnormal{The above definition differs slightly from that given in~\cite{DSS}, where statistical dispersion is defined by the condition $\lim_{t\to\infty}\langle(\Delta x)^2\rangle\ =+\infty$.}
\end{Remark}

Statistical dispersion is the weakest concept of dispersion among those introduced so far.
\begin{Proposition}\label{statweak}
If a regular mass distribution is totally or partially dispersive, then it is statistically dispersive. In particular, total dispersion implies that $\lim_{t\to\infty}\langle(\Delta x)^2\rangle =+\infty$.
\end{Proposition}
\begin{proof}
We prove first that total dispersion implies $\lim_{t\to\infty}\langle(\Delta x)^2\rangle =+\infty$. Fix $R>0$ arbitrarily and write
\begin{align*}
M&=\int_{\left|x-c_\rho(t)\right|\leq R}\rho\,dx+\int_{\left|x-c_\rho(t)\right|>R}\rho\,dx\\
&\leq\sup_{x_0\in\R^3}\int_{\left|x-x_0\right|\leq R}\rho\,dx+\int_{\left|x-c_\rho(t)\right|>R}\rho\,dx \:.
\end{align*}
Assume the solution is totally dispersive. Then by~\eqref{mgs} there exists $t_0=t_0(R)$ such that, for all $t>t_0$, 
\[
\sup_{x_0\in\R^3}\int_{\left|x-x_0\right|\leq R}\rho\,dx<\frac{M}{2}\:.
\]
Thus for $t>t_0$,
\[
\langle(\Delta x)^2\rangle\geq \frac{R^2}{M}\int_{\left|x-c_\rho(t)\right|>R}\rho\,dx\geq\frac{R^2}{2}\:,
\]
which yields the claim. To prove that partial dispersion implies statistical dispersion we use the dichotomy property of partially dispersive solutions, see Lemma~\ref{dichotomy}. Let $t_n,\rho_1^n,\rho_2^n$ satisfy the properties of Lemma~\ref{dichotomy}. Then
\begin{align*}
\langle(\Delta x)^2\rangle(t_n)&\geq \int_{\mathrm{supp}\rho_1^n}|x-c_\rho(t_n)|^2\rho_1^n\,dx+\int_{\mathrm{supp}\rho_2^n}|x-c_\rho(t_n)|^2\rho_2^n\,dx\\
&\geq d(c_\rho(t_n),\mathrm{supp}\rho_1^n)^2\|\rho_1^n\|_1+d(c_\rho(t_n),\mathrm{supp}\rho_2^n)^2\|\rho_2^n\|_1\\
&\geq C\left[d(c_\rho(t_n),\mathrm{supp}\rho_1^n)^2+d(c_\rho(t_n),\mathrm{supp}\rho_2^n)^2\right].
\end{align*}
By the triangle inequality,
\[
[\dots]\geq d(\mathrm{supp}\rho_1^n,\mathrm{supp}\rho_2n)^2\to\infty\:,\text{ as }t\to\infty\:,
\]
whence $\langle(\Delta x)^2\rangle(t_n)\xrightarrow{n}\infty$, which concludes the proof.
\end{proof}

\begin{Example}\label{shell-static}\textnormal{We give now an example of solution which is partially (and therefore statistically) dispersive but not strongly dispersive. 
This example is a modification of the fully dispersive shell considered before, in which a static, spherically symmetric configuration with given mass $M_0$ is located in the interior of a shell with mass $m$ (alternatively, the interior part may consist of a static shell \cite{Rein99}, leaving a neighborhood of the origin empty, or a spherically symmetric periodic solution, such as the one found by Kurth~\cite{Kurth}, see also next section). Since the potential inside the shell is constant, the static configuration in the interior will persist as long as the shell is moving outwardly. This again will be verified under condition~\eqref{initialw}, which now reads
$$
\inf\{w,w\in\textrm{ supp } f^0_\mathrm{shell}\}>\sqrt{\frac{M_0+m}{2\pi R_1(0)}}\:.
$$
Then we have 
$$
  (M_0+m) \langle(\Delta x)^2\rangle = \int_{\{R_1(t) >|x| \}} |x|^2 \rho\, dx + \int_{\{R_1(t) \le |x| \}} |x|^2 \rho\, dx \ge R_1(t)^2 m\:.
$$
By~\eqref{Rincrease}, this gives a growth of the spatial variance of order $t^2$. Partial dispersion also follows immediately by (\ref{Rincrease}).}
\end{Example}

%%%%%%%%%%%%%%%%%%%%%%%%%%%%%%%%%%
\section{Dispersion in the Vlasov-Poisson system}\label{applVP}
%%%%%%%%%%%%%%%%%%%%%%%%%%%%%%%%

Recall that strong dispersion implies total dispersion. Indeed we shall now prove that these two concepts of dispersion are equivalent for the Vlasov-Poisson system. Let
\[
E_{\mathrm{kin}}=\frac{1}{2}\int_{\R^6}|p|^2f\,dp\,dx\:,\quad E_{\mathrm{pot}}=\frac{1}{2}\int_{\R^3}|\nabla_x U|^2dx
\]
and note that
\begin{equation}\label{lowerboundekin}
E_{\mathrm{kin}}-\frac{Q^2}{2M}=\frac{1}{2}\int_{\R^6}\left|p-M^{-1}Q\right|^2f\,dp\,dx> 0\:.
\end{equation}
\begin{Proposition}\label{fullnec}
Let $f$ be a regular solution of the Vlasov-Poisson system. Then the following assertions are equivalent:
\begin{enumerate}
\item $f$ is strongly dispersive.
\item $f$ is totally dispersive.
\item $E_{\mathrm{pot}}\to 0$, as $t\to\infty$.
\end{enumerate}
Moreover, if any of the above holds, then $f$ satisfies the inequality
\[
E\geq \frac{Q^2}{2M}\:.
\]
\end{Proposition} 
%
%\begin{Corollary}
%Any solution with $E < \frac{Q^2}{2M}$ which is considered to be dispersive in some sense cannot be vanishing (e.g. not all mass is lost by infinity). 

%Then either it must be dichotomic or (despite how strange it might sound) compact ({\bf can we disregard this last possibility?}).
%\end{Corollary}
%
%\begin{Corollary}\label{tech2}
%{\bf The same result holds if we consider sequences of times instead (with the extended definitions). Furthermore, if there exists a sequence of times such that any of the above assertions holds, then $E\geq \frac{Q^2}{2M}$}.
%\end{Corollary}
%
\begin{proof}
By~\eqref{lowerboundekin}, 
\[
E>\frac{Q^2}{2M}-E_{\mathrm{pot}}\:,
\]
and the last claim follows by letting $t\to\infty$. Let us prove the equivalence of the three statements:

\emph{1. $\Longrightarrow$ 2.} Clear.

\emph{2. $\Longrightarrow$ 3}. Fix $R>0$ and rewrite the potential energy as $8 \pi E_{\mathrm{pot}}=I_1+I_2+I_3$, where
\begin{align*}
&I_1=\int\int_{|x-y|\leq1/R}\frac{\rho(t,x)\rho(t,y)}{|x-y|}\,dx\,dy\:,\\
&I_2=\int\int_{1/R<|x-y|\leq R}\frac{\rho(t,x)\rho(t,y)}{|x-y|}dx\,dy\:,\\
&I_3=\int\int_{|x-y|>R}\frac{\rho(t,x)\rho(t,y)}{|x-y|}dx\,dy\:.
\end{align*}
Using the Young inequality, see \cite{Lieb-Loss}, the first integral is bounded as
\[
I_1\leq C\|\rho(t)\|_{5/3}^2\left(\int_{|x|\leq R^{-1}}|x|^{-5/4}dx\right)^{4/5}\leq C R^{-7/5}\:;
\]
we recall that $\|\rho(t)\|_{5/3}\leq C$ for regular solutions of the Vlasov-Poisson system, see for instance~\cite{ReinLibro}. For $I_3$ we use that 
\[
I_3 \leq \frac{M^2}{R}\:.
\]
Finally
\[
I_2\leq R\int_{|x-y|\leq R}\rho(t,x)\rho(t,y)\,dx\,dy\leq MR\sup_{y\in\R^3}\int_{|x-y|\leq R}\rho(t,x)\,dx=R\,\varepsilon_R(t)\:,
\]
where, by~\eqref{mgs}, $\varepsilon_R(t)\to 0$, as $t\to\infty$, for all $R>0$. Collecting,
\[
8 \pi E_{\mathrm{pot}}\leq C\left( R^{-1}+R^{-7/5}\right)+R\,\varepsilon_R(t)\:.
\]
Taking the limit $t\to\infty$ and then $R\to\infty$ concludes the proof.

\emph{3. $\Longrightarrow$ 1.} We recall from \cite{IR} the following interpolation inequality
\begin{equation}
\label{lines_interpolation}
\|\rho\|_{5/3}^{5/3} \le C t^{-2} \left(\int_{\R^6} |x-tp|^2 f(t,x,p)\, dx\,dp \right)
\end{equation}
and the pseudoconformal law for the attractive case
\begin{equation}
\label{BPR}
\frac{d}{dt}\int_{\R^6} |x-tp|^2 f(t,x,p)\, dx\,dp =  \frac{d}{dt}\left(t^2 \|\nabla_x U(t) \|_2^2\right) - t \|\nabla_x U(t) \|_2^2. 
\end{equation}
Integrating (\ref{BPR}), we get
$$
\int_{\R^6}|x-t p|^2 f\,dx\,dp - \int_{\R^6} |x|^2 f^0\,dx\,dp = t^2 \|\nabla_x U(t) \|_2^2 - \int_0^t s \|\nabla_x U(s) \|_2^2 \ ds,
$$
so that
$$
   0 \le t^{-2} \int_{\R^6}|x-tp|^2 f \,dx\,dp \le t^{-2} \int_{\R^6}|x|^2 f^0\,dx\,dp + \|\nabla_x U(t) \|_2^2\:,
$$
and the r.h.s. converges to zero by hypothesis, which in combination with (\ref{lines_interpolation}) concludes the proof.
\end{proof}
%
%\begin{Remark}\textnormal{
%{\bf The proof yields more precise information. Under the given assumptions, we get}}
%\begin{itemize}
%\item for $R$ fixed, $\lim_{t \to \infty} \sup_{x_0 \in \R^3} \int_{|x-x_0|<R} \rho(t,x) \, dx \le O(\|\rho\|_p) \ \forall 1\le p\le \infty$.
%%
%\item $|E_{\mathrm{pot}}| \le O\left(\|\rho\|_p^{\frac{p'}{3+2p'}}\right) \ \forall 1< p\le \infty$.
%%
%\item $\|\rho\|_{5/3} \le O(|E_{\mathrm{pot}}\|^{3/5})$.
%%
%\end{itemize} 
%\end{Remark}
%%
%\begin{Conjecture}
%{\bf The higher the symmetry of the initial datum -measured in terms of the size of the group of transformations that leave it invariant-, the faster it disperses strongly, whenever this makes sense}.
%\end{Conjecture}
%%
%\begin{Conjecture}
%{\bf For the spherically symmetric class of initial data}, if we are in the hypothesis of Proposition \ref{fullnec} and $E>\frac{Q^2}{2 M}$, then 
%\begin{enumerate}
%\item the dispersion rates for $|E_{\mathrm{pot}}|$ and $\|\rho\|_p$ are universal/generic: they do depend only on the degree of integrability of $\rho_0$ beyond $L^{5/3}(\R^3)$ -the higher the exponent, the faster the rates. Furthemore, the radius of the system grows linearly.

%\item if $\rho_0 \in L^\infty(\R^3)$ then $|E_{\mathrm{pot}}| \sim t^{-1}$ and $\|\rho\|_p \sim t^\frac{3-3p}{p}$.

%\end {enumerate}
%\end{Conjecture}

In \cite{DSS} a sufficient condition for statistical dispersion was established, which we reprove here in a slightly different way. We shall use the identity
\begin{equation}\label{deltaxid}
M\frac{d^2}{dt^2}\langle(\Delta x)^2\rangle=2E+2E_\mathrm{kin}-2\frac{Q^2}{M}\:,
\end{equation}
which is proved by direct calculation.
\begin{Proposition}\label{partialsuf}
Solutions of Vlasov-Poisson which satisfy the condition $E>\frac{Q^2}{2M}$
%\begin{itemize}
%\item[(i)]$E>\frac{Q^2}{2M}$, or
%\item[(ii)]$E=\frac{Q^2}{2M}$ and $\int_{\R^6}x\cdot (p-M^{-1}Q)f^0(x,p)\,dx\,dp>0$,
%\end{itemize}
are statistically dispersive. In particular, there exists constants $C_1,C_2>0$ such that, for all sufficientely large times, 
\begin{equation}\label{stadispEpos}
C_1t^2\leq\langle(\Delta x)^2\rangle\leq C_2t^2\:.
\end{equation}
\end{Proposition}
\begin{proof} First we rewrite~\eqref{deltaxid} as
\begin{equation}\label{deltaxid2}
M\frac{d^2}{dt^2}\langle(\Delta x)^2\rangle=4E-2\frac{Q^2}{M}+2E_{\rm pot}\:,
\end{equation}
Using that $E_{\rm pot}\geq 0$ and integrating in time twice we get
\[
\langle(\Delta x)^2\rangle(t)\geq \langle(\Delta x)^2\rangle(0)+\left[\frac{d}{dt}\langle(\Delta x)^2\rangle\right]_{t=0}t+\frac{2}{M}\left(E-\frac{Q^2}{2M}\right)t^2\:,
\]
where 
\begin{equation}\label{auxiliary1}
\left[\frac{d}{dt}\langle(\Delta x)^2\rangle\right]_{t=0}=\frac{2}{M}\int_{\R^6}x\cdot (p-M^{-1}Q)f^0(x,p)\,dx\,dp\:.
\end{equation}
The bound from below follows immediately. To prove the upper bound, we recall---see~\cite{DSS,ReinLibro} for instance---that the potential energy satisfies the bound
\[
E_\mathrm{pot}\leq C\sqrt{E_\mathrm{kin}}\:,
\]
where the positive constant $C$ depends only on $M=\|f(t)\|_1$ and $\|f(t)\|_\infty=\|f^0\|_\infty$. Thus
\[
E_\mathrm{kin}-C\sqrt{E_\mathrm{kin}}-E\leq 0\:,
\]
which in the case of non-negative total energy $E$ gives immediately a uniform upper bound on the kinetic energy:
\begin{equation}
E_\mathrm{kin}\leq \left(\frac{1}{2}\left(C+\sqrt{C^2+4E}\right)\right)^2
\end{equation}
and therefore the potential energy is uniformly bounded as well.
Using this in~\eqref{deltaxid} it follows that $\langle(\Delta x)^2\rangle\leq C_2t^2$ and the proof is complete.
\end{proof}
The threshold $Q^2/(2M)$ represents the kinetic energy of a point at the center of mass having the same mass of the whole system. The case $E=Q^2/(2M)$ is settled in the following proposition.
\begin{Proposition}
\label{partialsuf2}
Solutions of Vlasov-Poisson which satisfy $E=\frac{Q^2}{2 M}$ are statistically dispersive. 
%In particular, there exists a positive constant $C$ such that $\langle(\Delta x)^2\rangle \geq C\,t$, for all sufficiently large times. 
%{\bf Furthermore...} (rodeo sin hablar del virial) if such solution is not vanishing, $\langle(\Delta x)^2\rangle = O(t^2).$
%{\bf If $E=\frac{Q^2}{2 M}$ then there exists a sequence of times such that statistical dispersion holds along it -Remark \ref{tech}. More precisely, either the solution is statistically dispersive -in the usual sense- or along the aforementioned sequence the solution is strongly dispersive.}
\end{Proposition}
\begin{proof}
After a Galilean transformation we may assume $Q=0$. In this frame, the solutions under consideration have zero energy. Thus we compute
%We remark first that as the $L^q$-norms of $f(t)$ are conserved any clustering point ({\bf Qu\'e conio es un clustering point?}) of $f(t)$ as $t$ goes to infinity belongs to $L^q(\R^6)$; and in particular no such cluster point can be a measure. Thus $\rho(t)$ cannot shrink into a point as $t$ grows big. This shows that $\langle(\Delta x)^2\rangle$ has a lower bound $C>0$ (which depends on the initial condition of course). 
$$
\frac{M}{2}  \frac{d}{dt} \langle(\Delta x)^2\rangle = \int_\mathrm{\R^6} (x\cdot p)f^0 \, dxdp +\int_0^t E_\mathrm{kin}(\tau)\, d\tau.
$$
The claim is obvious if $\int (x\cdot p)f^0 \, dxdp\geq 0$. More exactly, in the latter case we have $\langle(\Delta x)^2\rangle \geq C\,t$, for a positive constant $C$. Otherwise there exist two possibilities:
\begin{enumerate}

\item There exists $t^*>0$ such that $- \int (x\cdot p)f^0 \, dxdp =  \int_0^{t^*} E_\mathrm{kin}(\tau)\, d\tau$.  In this case too $\langle(\Delta x)^2\rangle \geq C\,t$.

\item There holds the bound $ \int_0^\infty E_\mathrm{kin}(\tau)\, d\tau\leq  -\int (x\cdot v)f^0 \, dxdv$. 
%Then $\langle(\Delta x)^2\rangle$ decreases in a monotone way to a strictly positive limit. 
Hence there exists a sequence of times $t_n\xrightarrow{n}\infty$ such that $E_\mathrm{kin}(t_n)\xrightarrow{n} 0$, and thus $E_{\mathrm{pot}}(t_n)\xrightarrow{n} 0$ as well. From the proof of the implication $3\Rightarrow 1$ in Proposition \ref{fullnec}, we see that $\|\rho(t_n)\|_{5/3}\xrightarrow{n}0$. In particular
\[
\lim_{n\to\infty}\sup_{x_0\in\R^3}\int_{|x-x_0|\leq R}\rho(t_n,x)\,dx=0\:,\quad\forall\,R>0
\]
and repeating the argument at the beginning of the proof of Proposition~\ref{statweak}, we conclude that  $\langle(\Delta x)^2\rangle(t_n)\xrightarrow{n}+\infty$.
\end{enumerate}
\end{proof}

\begin{Example}\label{shell-static-neg}\textnormal{At this point it is interesting to reconsider the example of the shell surrounding a static configuration introduced at the end of Section~\ref{dispdef}. Say that the inner part of the solution has mass $M_0$ and that the surrounding shell has mass $m$ and initial inner radius $R_1$. Our previous computations show that the escape threshold associated with this configuration is 
$$
 \sqrt{\frac{M_0+m}{2 \pi R_1}}\:,
$$
see~\eqref{initialw},
so that any particle with initial radial momentum greater than this threshold will escape towards infinity. We shall now prove that it is possible to obtain an escaping shell even when the total energy of the system is negative. Note that the total energy $E$ consists of the energy $E_0$ of the interior part plus the kinetic energy of the shell minus the potential energy of the shell (the   interaction energy has negative sign and as a consequence  the term contributing to the potential energy can be not consired). Neglecting the last negative term and estimating above the kinetic energy of the shell we get
$$
E< E_0 + \frac{1}{2}m \sup_{\mathrm{shell}} |p|^2\:,
$$
where the supremum is taken in the support of the shell at time $t=0$.
The interior part has energy $E_0<0$, since it is static (cf. Proposition \ref{generic}); thus in order to have the whole shell escaping to infinity while the total energy of the system remains negative we must choose the initial radial momenta for all the particles in the shell according to 
$$
   \frac{M_0+m}{2 \pi R_1} < \inf_{\mathrm{shell}} w^2 < \inf_{\mathrm{shell}} |p|^2 < \sup_{\mathrm{shell}} |p|^2 < \frac{-2 E_0}{m}.
$$
This can be done if $m$ is strictly contained in the interval between zero and the value $\frac{1}{2}[-M_0 + \sqrt{M_0^2 - 16 \pi E_0R_1}]$. For bigger values of $m$ we have no guarantee that the total energy can be kept negative.}
\end{Example} 
%Eventually aising the value of m and keeping in any case the radial velocities of the particles of the shell above the escape threshold we shall get configurations with zero and positive energy which show partial disperssion with precisely an amount $m$ of lost mass $\mathcal{M}_{out}$.
The previous example shows us that:

\begin{itemize}

\item There are solutions that are partially (therefore statistically) dispersive with $E<Q^2/2M$, so that there is no simple way to extend the results of Propositions \ref{partialsuf} and \ref{partialsuf2}. By Proposition~\ref{fullnec}, these solutions cannot be totally dispersive.

\item We know that spherically symmetric solutions with $E>0$ disperse statistically with a dispersion rate of $t^2$ (equivalently their spatial support spreads with a velocity of order $t$). We will see in Section~\ref{kurthsec} that for solutions with $E=0$ this needs not to be the case. On the other hand the previous example shows that there are solutions with $E\le 0$ which also statistically disperse with a rate $t^2$. So there is no evident relation between the admisible rates of dispersion and the sign of the energy.

%\item there are examples of solutions with $E\le 0$ (and obviously with $E>0$) which are not virialized (see forthcoming sections...).

\item There is no lower limit for the fraction of total mass of the system that is lost to infinity for a partially dispersive system. 
%{\bf We might conjecture that if $E<0$ there is an upper limit for this quantity in terms of several macroscopic parameters} (say, for instance, that no more than a half of the mass of the system can be dissapearing in this regime). Obviously without energy restrictions there is no value for the fraction of mass that escapes that is forbidden, as can be seen for the examples of shell surrounding static configuration that are allowed to have positive energy.

%\item presumably the integrability properties of $\rho^0$ play no role in the rates of (statistical) dispersion. For instance we might construct shells with initially unbounded associated density and the argument sketched above goes through without noticing this particular feature.
\item Using these ideas we can also show that there exist dispersive  
solutions as ``close" as desired to a stable steady state. More  
precisely, given a polytropic steady state  
(see \cite{SS} for details and notation) with mass $M$,  
polytropic index $\mu$ and $L^{1+ 1/\mu}$-norm $J$, which is stable in  
the sense of \cite{SS}, we can find solutions as described above that are partially dispersive and that remain in the stability region of the polytrope.  (This is not  a  
contradiction since the quantity of mass that is lost to infinity is almost  
negligible.) These solutions correspond to initial data of the following type. Starting form the given polytrope we construct---by scaling---a second polytrope 
having mass $M-m$ and   
$L^{1+ 1/\mu}$-norm $J-j$, for $m,j$ positive and small. Then we add an outer shell of mass $m$ and $L^{1+1/ \mu}$-norm less or equal to  
$j$. We choose $m,j$ in order that all the computations  
in Example \ref{shell-static-neg}  remain valid and that the total energy of the solution is as close  
as desired to the energy of the original polytrope. Then we can invoke the stability  
criterium in \cite{SS}. 
\end{itemize}

\subsection{Kurth's solution}\label{kurthsec}
%{\bf 3. Kurth's example.}
%{\bf We take Poisson's equation as $\Delta \mathcal{U} = 4 \pi  \rho$, apparently.}
We conclude this section by considering an explicit class of  spherically symmetric solutions to the Vlasov-Poisson system found by Kurth in \cite{Kurth} (see also \cite{Horst84}, \cite{ReinLibro}). The main idea of this model is to find a solution whose associated density is of the form
\begin{equation}\label{rhokurth}
\rho(t,x)= (4 \pi /3)^{-1} \phi(t)^{-3}\chi_{\{|x| < \phi(t)\}}
\end{equation}
where the function $\phi$ is interpreted as the radius of the system. This function $\phi$ solves the following ODE
$$
   \phi^3 \phi'' + \phi =1\:,
 $$
 subject to the initial condition  $\phi(0)=1\:$.
We get solutions to the Vlasov-Poisson system exhibiting different types of behavior depending on the prescribed value of $\phi'(0)$:
\begin{itemize}
\item If $\phi'(0)=0$, the solution is static;
\item If $0<|\phi'(0)|<1$, the solution is periodic in time;
\item If $|\phi'(0)|\ge 1$, the solution is strongly dispersive.
\end{itemize}

%It is given by 
%\begin{equation}\label{kurth}
%   f(t,x,p)= \frac{3}{4 \pi^3}\chi_{\{|x \wedge p|<1\}} \left(1 - \Big(\frac{x}{\phi(t)}\Big)^2 - (\phi(t) p - \phi'(t) x)^2 + (x \wedge p)^2 \right)_+^{-1/2},
%\end{equation}
%where the function $\phi$ solves 
%$$
%   \phi^3 \phi'' + \phi =1,\, \phi(0)=1\:,
%$$
%and $\phi'(0)$ is given. The associated density of the models is 
%\begin{equation}\label{rhokurth}
%\rho(t,x)= (4 \pi /3)^{-1} \phi(t)^{-3}\chi_{\{|x| < \phi(t)\}}
%\end{equation}
%and thus the function $\phi$ is interpreted as the radius of the system. We get different types of behavior depending on the value of $\phi'(0)$:
%%
%\begin{itemize}
%\item If $\phi'(0)=0$, the solution is static.
%\item If $0<|\phi'(0)|<1$, the solution is periodic in time.
%\item If $|\phi'(0)|\ge 1$, the solution is strongly dispersive
%\end{itemize}
%%In this example we can show that modulating that parameter the associated family of solutions ranges between static states, passing trough breathing modes towards strongly dispersive solutions.%even in the case where the energy is zero $E=0$ which was not previously cover by (\ref{hipo1}). The section is based on some of the computations performed in the final part of ().
%The energy of this class of solutions satisfies $E \in [-3/5,\infty[$. 

Let us relate the above classification in terms of $\phi'(0)$ with the values of the energy $E$. Note first that the associated distribution function $f$ can be chosen to be spherically symmetric so that $Q=0$. Doing so, the energy of Kurth's solutions is given by, see \cite{Horst84}, 
\[
E=\frac{3}{5}\left((\phi')^2+\phi^{-2}-2\phi^{-1}\right)=\frac{3}{5}(\phi'(0)^2-1)\:.
\]
Moreover $M=1$, which follows integrating~\eqref{rhokurth}. Thus
\begin{itemize}
\item If $E=-3/5$ ($\Leftrightarrow\phi'(0)=0$), the solution is static (this steady state is studied in \cite{BFH}).
\item If $-3/5 < E <0$ ($\Leftrightarrow 0<|\phi'(0)|<1$), the solution is time periodic. 
%{\bf En alg\'un sentido, la soluci\'on estacionaria de Kurth es inestable}.
\item If $E\ge 0$  ($\Leftrightarrow |\phi'(0)|\geq 1$), the radius of the system goes to infinity and, by~\eqref{rhokurth}, $\rho\to 0$ in $L^q$, for all $q>1$, {\it i.e.}, the solution is strongly (and therefore also totally) dispersive. When $E=0$ ($\Leftrightarrow |\phi'(0)|=1$), we have $\langle (\Delta x)^2\rangle \sim t^{4/3}$. When $E >0$ ($\Leftrightarrow |\phi'(0)|>1$), we get $\langle (\Delta x)^2\rangle \sim t^2$, in agreement with Proposition~\ref{partialsuf}.
\end{itemize}
Let us prove the latter claim. First we show that $\langle (\Delta x)^2\rangle \sim O(\phi(t)^2)$. Since the solution under study is spherically symmetric, we have
$$
  \langle (\Delta x)^2\rangle = \int_{\R^3} |x|^2 \rho(x) \, dx = 4 \pi \left(\frac{4 \pi}{3}\right)^{-1} \phi(t)^{-3} \int_0^{\phi(t)} r^4 \, dr = \frac{3}{5} \phi(t)^2.
$$
Thus we reduce the problem to find out the large time behavior of the function $\phi$. Following Kurth~\cite{Kurth}, if $|\phi'(0)|=1$ we have that 
$$
\phi(t) = \frac{1}{2} (1 + v(t)^2)\:,
$$ where $v(t)$ solves
$$
   v(t) + \frac{1}{3} v(t)^3 = 2\left(t + \frac{2}{3} \right).
$$
For $t$ big the term $v(t)^3$ dominates and then $\phi(t) \sim O(t^{2/3})$.
If $|\phi'(0)|>1$, we have that 
$$
\phi(t) = \frac{|\phi'(0)| \mbox{ch}\, (v(t)) - 1}{|\phi'(0)|^2 -1}\:,
$$ where $v(t)$ solves
$$
   v(t) - |\phi'(0)| \mbox{sh}\, (v(t)) = (|\phi'(0)| -1)^{3/2} (t-t_0)
$$
($t_0$ depends on $|\phi'(0)|$). For $t$ big enough $|\mbox{sh}\, (v(t))|$ dominates $|v(t)|$, and we infer that $|v(t)| \sim O(\mbox{log}\, t)$, which entails $\phi(t) \sim O(t)$.

Note that the solution with $E=0$ (there are two of them actually) is the only known example of statistically dispersive solution for which statistical dispersion takes place at a slower rate than $t^2$ and for which the support spreads more slowly than $t$. Also the rate for strong dispersion is slower than for the other examples considered in this paper. Taking a closer look to the trajectories reveals that these are strongly oscillatory (like in a forced harmonic oscillator). 

This special solution highlights also the role of condition $d)$ in Lemma \ref{dichotomy}. For we can surround this ``slowly dispersing'' solution with a strongly dispersive shell configuration, and the resulting solution verifies $a),\ b)$ and $c)$ but not $d)$, and happens to be totally but not partially dispersive, see previous Remark \ref{rem2}.

\section{Remarks on non-dispersive solutions}\label{nondispersive}

In this section we briefly consider non-dispersive solutions of the Vlasov-Poisson system. In particular, we consider steady state solutions and time periodic (breather) solutions.  

We distinguish between two types of steady states: Static solutions, which are defined as time-independent solutions of the Vlasov-Poisson systen~\eqref{rhoVP}--\eqref{vlasovVP}, and travelling steady states, which are defined as solutions $f(t,x,p)$ such that there exists a Galilean transformation $\mathcal{G}_u$ such that $f\circ\mathcal{G}_u$ is a time independent solution of the Vlasov-Poisson system ({\it i.e.}, a static solution). For static solutions one has $Q=0$, wheras $Q\neq 0$ for travelling steady states. In fact, it is easy to show that the Galilean transformation which transforms a travelling steady state with linear momentum $Q$ and mass $M$ to a static solution (with linear momentum $Q=0$ and mass $M$) is the transformation $\mathcal{G}_u$, where
\[
u=\frac{Q}{M}\:.
\]
%For the existence of steady states (i.e., time independent) solutions, the necessary condition on the energy is that 
%\begin{equation}\label{negativeenergyVP}
%E<0\:.
%\end{equation}

The construction of static solutions for the Vlasov-Poisson system can be done in two ways. Firstly, by choosing a particle density $f$ which depends only on quantities that are conserved along the characteristics of the time independent Vlasov equation; with this choice the Vlasov equation is automatically satisfied and the problem is reduced to that of proving an existence theorem for the non-linear elliptic equation obtained by replacing $f$ in the Poisson equation.  So far this `direct' method was used mostly in the spherically symmetric case (see however~\cite{Schulze}), where by the Jeans Theorem~\cite{BFH} all solutions of the time independent Vlasov equation can be expressed in terms of the particles energy and angular momentum. A second method to construct steady states of the Vlasov-Poisson system is by minimizing the energy (or a related) functional subject to suitable constraints (the choice of the functional and/or the constraints selects the type of steady state to be constructed). The advantage of the variational method on the direct method is that the former automatically proves a stability property for the steady state. We refer~\cite{Guo03, Rein2000,Rein2000bis, SS, ReinLibro} and references therein for several works on the construction and stability of steady states of the Vlasov-Poisson system.

%\noindent
To make a connection with the results proved in the previous sections for dispersive solutions, let us note the following simple bounds on the energy of steady states and time periodic solutions:
\begin{Proposition}
\label{generic}
Static solutions of the Vlasov-Poisson system satisfy $E<0$, whereas travelling steady states satisfy $E<\frac{Q^2}{2M}$.
Time periodic solutions of the Vlasov-Poisson system satisfy $E<-\frac{Q^2}{2M}$. 
\end{Proposition}
\begin{proof}
The key to prove the result is the dilation identity:
\begin{equation}
\label{dilation_identity}
\frac{d}{dt} \int_{\R^6}x\cdot p\, f\,dp\,dx=E+E_{\mathrm{kin}}\:,
\end{equation}
which follows by direct computation.
If $f$ is a static solution, then the previous identity implies the {\it virial} relation $E=-E_{\mathrm{kin}}$,  which yields the claim for static solutions. For travelling steady states, a Galilean transformation with $u=(M)^{-1}Q$ yields the virial identity in the form $E-(M)^{-1}Q^2=-E_{\mathrm{kin}}$, which, using~\eqref{lowerboundekin}, implies the claim of the proposition for travelling steady states. In the case of periodic solutions, we integrate the dilation identity
over a period $T$ to get
\[
0=E T+\int_0^TE_{\mathrm{kin}} dt\:.
\]
Using again~\eqref{lowerboundekin} concludes the proof.
\end{proof}
%
%\begin{Remark}
%({\bf Otra manera de presentar algunos c\'alculos}). The Galilean invariance -rotational invariance, etc- of $\langle(\Delta x)^2\rangle$ and the computation (\ref{auxiliary1}) show that for any steady state we must have
%$$
%\int \int (x\cdot p) f^0(x,p)\ dxdp = \frac{1}{M} Q \cdot \int \int x f^0(x,p)\ dxdp
%$$
%({\bf esto parece muy tonto, pero no lo he visto establecido de forma expl\'{\i}cita en ning\'un sitio}). Combined with (\ref{first_derivative}) we obtain the virial relation
%$$
%   E + E_{\mathrm{kin}} = \frac{Q^2}{M} \ a.e. \ t
%$$
%\end{Remark}
%
\begin{Remark}\textnormal{
In the case of static solutions of the Nordstr\"om-Vlasov system, the inequality $E<0$ is replaced by $E<M$, see~\cite{CSS}. For the (spherically symmetric) Einstein-Vlasov system, one can obtain an inequality for static solutions which involves not only the energy and the mass but also the central redshift~\cite{unp}.}
\end{Remark}

%%%%%%%%%%%%%
\subsection{Virialized  solutions}\label{virializedsols}
It is a classical result in Astrophysics that bounded systems of self-gravitating particles roughly in equilibrium verify that the time average of the kinetic energy of the ensemble equals twice the time average of its potential energy. This statement and some of its variants and particularizations go under the name of ``virial theorems'' (cf. \cite{Pollard,Saslaw,Binney}, for instance), and are common tools in Astrophysics. We shall comment here on the connection between the notion of {\it virialized} solutions of the Vlasov-Poisson system and our preceding results.
 
In this section we shall only consider solutions such that $Q=0$, {\it i.e.}, the reference frame is chosen at rest with respect to the center of mass of the system, as we did in the proof of Proposition \ref{partialsuf2}. Following \cite{Pollard} we shall say that a solution of the Vlasov-Poisson system is virialized if and only if
\begin{equation}
\label{virial_limit}
  \lim_{t\to \infty} \frac{\int_0^t E + E_\mathrm{kin}(\tau) \, d\tau}{t} =0\:.
\end{equation}
Note that solutions with $E>0$ cannot be virialized in this sense. In fact, it is a straightforward consequence of the inequality~\eqref{lowerboundekin} that virialized solutions of the Vlasov-Poisson system must necessarily satisfy
\[
E\leq 0\:.
\]
Examples of virialized solutions are Kurth's solutions with energy $E\leq 0$. 
\begin{Remark}\textnormal{
The notion of virialized system is usually applied in Astrophysics only in the case of $N$-body bounded systems, but, as pointed out in \cite{Pollard}, strict boundedness is not necessary (Kurth's solution with energy $E=0$ shows that also in the case of Vlasov-Poisson, the support of a virialized solution can spread out to infinity). This remark led to interpret the virialized systems as  apparently bounded systems, {\it i.e.}, systems that disperse so slowly that in our time scale %
%-which amounts for only a few snapshots of the whole movie
they appear to be bounded and in equilibrium.}
\end{Remark} 
The following proposition extends the result in \cite{Pollard}, which is valid for $N$-body systems, to the continuos setting; note that the diameter of the $N$-body system used in ~\cite{Pollard}, is replaced by the statistical dispersion operator $\langle(\Delta x)^2\rangle$.

\begin{Lemma}
\label{virialequivalence}
 Let $f(t)$ be a given solution of the Vlasov-Poisson system with $Q=0$. Then the following statements hold true:
 \begin{enumerate}
 
 \item If $f(t)$ is virialized then $\lim_{t \to \infty} \frac{\langle(\Delta x)^2\rangle}{t^2}=0$.
 
\item If $\lim_{t \to \infty} \frac{\langle(\Delta x)^2\rangle}{t^2}=0$ and $\lim_{t\to \infty} \frac{\int_0^t E + E_\mathrm{kin}(\tau) \, d\tau}{t}$ exits, then $f(t)$ is virialized.

 \end{enumerate}
\end{Lemma}
\begin{proof}
The proof is a straightforward application of L'H\^opital's rule  (as formulated in \cite{Rudin}) and  of\eqref{deltaxid}. 
\end{proof}
\begin{Remark}\textnormal{The existence of the limit in (\ref{virial_limit}) can be guaranteed under some of the most frequent situations: Spherical symmetry \cite{Batt}, periodicity in time (including static solutions of course) and whenever $E_{\mathrm{kin}}(t)$ has a limit as $t\to \infty$.}
\end{Remark}

\section{Summary and open problems}
The main results of this paper are summarized in Table~\ref{table}. Completing the entries marked with a question mark would lead to a considerable extension of the result presented here and of our understanding of the large time behavior of the Vlasov-Poisson system in the gravitational case. 
\begin{table}
\begin{center}
\begin{tabular}{|c|ccc|}
\hline  & & & \\[-2ex]
{\bf Dispersive behavior} & {\bf Necessary}  & {\bf Sufficient} & {\bf Example}  \\ \hline & & & \\[-2ex] 
Strong Dispersion & $E\geq Q^2/2M$ & ? & Example~\ref{shellexample} \\
{\tiny (VP)}${\uparrow}$ $\downarrow$ & & &\\
Total Dispersion & $E\geq Q^2/2M$ & ? & Kurth's $E\geq 0$ \\
$\downarrow$ & & &\\
Statistical Dispersion &  ? & $E\geq Q^2/2M$ &  \\
$\uparrow$ & & &\\
Partial Dispersion &  ? & ? & Example~\ref{shell-static} \\[1ex]
{\bf Other solutions} & && \\[1ex] 
Static Solutions &  $E<0$ & -- & Kurth's $E=-\frac{3}{5}$ \\
Periodic Solutions&  $E<-Q^2/2M$ &--& Kurth's $E\in (-\frac{3}{5}, 0)$\\
Virialized Solutions& $E\leq 0$ &--& Kurth's $E=0$  \\ 
\hline
\end{tabular}
\caption{Main results proved in the paper and open problems.}
\label{table}
\end{center}
\end{table}

%%%%%%%%%%%%%%%%%%%%%%%%%%%%%%%%%%

\end{document}